\documentclass[reqno,10pt]{amsart}
\usepackage{amsmath}
\usepackage{amsthm}
\usepackage{amssymb}
\usepackage[mathscr]{euscript} 
\usepackage{bm}
\usepackage{pgf}
\usepackage{color}
\usepackage{graphicx}
\usepackage{comment}
\usepackage{mathtools}
\usepackage{soul}

\newcommand\be[1]{\begin{equation}\label{#1}}
\newcommand\ee{\end{equation}}
\newcommand\ba[1]{\begin{align}\label{#1}}
\newcommand\ea{\end{align}}
\newcommand\bas{\begin{align*}}
\newcommand\eas{\end{align*}}
\newcommand\nn{\nonumber}

\newtheorem{theorem}{Theorem}[section]
\newtheorem{lemma}[theorem]{Lemma}

\theoremstyle{definition}
\begingroup
\newtheorem{definition}[theorem]{Definition}
\newtheorem{remark}[theorem]{Remark}

\endgroup

\newcommand{\Rz}{\mathbb{R}}
\newcommand{\Nz}{\mathbb{N}}
\newcommand{\Zz}{\mathbb{Z}}

\newcommand{\eps}{\varepsilon}

\newcommand{\eone}{\boldsymbol e_1}
\newcommand{\etwo}{\boldsymbol e_2}
\newcommand{\ethree}{\boldsymbol e_3}
\newcommand{\ei}{\boldsymbol e_i}

\newcommand{\EEE}{\color{black}}
\newcommand{\UUU}{\color{black}}
\newcommand{\PPP}{\color{black}}

\usepackage{metalogo}
\usepackage{booktabs}
\usepackage{hyperref}
\makeatletter
\hypersetup{%
     pdfpagemode={UseOutlines},
     bookmarksopen,
     pdfstartview={FitH},
     colorlinks,
     linkcolor={blue},
     citecolor={red},
    urlcolor={red}
  }

\usepackage{color}
\newcommand{\cb}{\color{blue}}

\newcommand{\cm}{\color{magenta}}

\newenvironment{proofad1}{\removelastskip\par\medskip   
\noindent{\em Proof of \rm Theorem \ref{main1}.}
\rm}{\penalty-20\null\hfill$\square$\par\medbreak} 

\newenvironment{proofad2}{\removelastskip\par\medskip   
\noindent{\em Proof of \rm Theorem \ref{main2}.}
\rm}{\penalty-20\null\hfill$\square$\par\medbreak} 

\begin{document}

\title[$N^{3/4}$ law in the cubic lattice]{$N^{3/4}$ law in the cubic lattice}

\author{Edoardo Mainini}
\address[Edoardo Mainini]{Dipartimento di Ingegneria meccanica, energetica, gestionale e dei trasporti, 
  Universit\`a  degli studi di Genova, Via all'Opera Pia, 15 - 16145 Genova Italy.}
\email{mainini@dime.unige.it}

\author{Paolo Piovano}
\address[Paolo Piovano]{Faculty of Mathematics, University of Vienna,   Oskar-Morgenstern-Platz 1, A-1090 Vienna, Austria}
\email{paolo.piovano@univie.ac.at}
\urladdr{http://www.mat.univie.ac.at/$\sim$piovano/}

\author{Bernd Schmidt}
\address[Bernd Schmidt]{Institute of Mathematics, University of Augsburg,   Universit\"atsstr. 14, 86159 Augsburg, Germany}
\email{bernd.schmidt@math.uni-augsburg.de}
\urladdr{http://www.math.uni-augsburg.de/ana/schmidt.html}

\author{Ulisse Stefanelli}
\address[Ulisse Stefanelli]{Faculty of Mathematics, University of
  Vienna,   Oskar-Morgenstern-Platz 1, A-1090 Vienna, Austria and
  Istituto di Matematica Applicata e Tecnologie Informatiche {\it
    E. Magenes} - CNR, via Ferrata 1, I-27100 Pavia, Italy}
\email{ulisse.stefanelli@univie.ac.at}
\urladdr{http://www.mat.univie.ac.at/$\sim$stefanelli}


\begin{abstract}
We investigate the Edge-Isoperimetric Problem (EIP) for sets with $n$
elements of the cubic lattice  by emphasizing its relation with the
emergence of the Wulff shape in the crystallization
problem. Minimizers $M_n$ of the edge perimeter are shown to deviate
from a corresponding cubic Wulff configuration with respect to their
symmetric difference by at most \UUU $ {\rm O}(n^{3/4})$ \EEE
elements. \UUU The \EEE exponent $3/4$ \UUU is optimal. This 
 extends to the cubic lattice analogous results that have already  been 
established   for the triangular, the hexagonal, and
the square lattice in two space dimensions. \EEE 
\end{abstract}


\maketitle

\pagestyle{myheadings}

\section{Introduction}
In this contribution we consider the {\it Edge-Isoperimetric Problem} (EIP) in 
$$\Zz^3:=\{k_1 \eone+k_2\etwo+ k_3 \ethree\,\,:\quad\textrm{$k_i\in\Zz$ for $i=1,2,3$}\}$$
where  $\eone:=\left(1,0,0\right)$, $\etwo:=(0,1,0)$ and  $\ethree:=(0,0,1)$.
For any set $C_n$ made of $n$  elements of $\Zz^3$, we denote by $\Theta(C_n)$ the {\it edge boundary} of $C_n$, i.e.,  \be{edgeboundary}
\Theta(C_n):=\{ (x,y)\in \Zz^3\times\Zz^3 \ : \ \textrm{$|x-y|=1$, $x\in C_n$ and $y\in \Zz^3\EEE\setminus C_n$} \}
\ee 
and we refer to its cardinality $\#\Theta(C_n)$  as the {\it edge perimeter} of $C_n$.
Given the family $\mathcal{C}_n$ of all sets $C_n \subset \Zz^3$ with $n$ elements, the Edge-Isoperimetric Problem over $\mathcal{C}_n$ consists in considering the minimum problem
\begin{equation}\label{eip3}
\theta_n:=\min_{C_n\in\mathcal{C}_n} \#\Theta(C_n),
\end{equation}
which we denote by EIP$_n$, and in characterizing the EIP$_n$ solutions. 
The EIP is a classical combinatorial problem and a review on the results in Combinatorics can be found in \cite{Bezrukov, Harper}. Beyond its relevance in pure combinatorics, the EIP (and corresponding problems for similar notions of perimeter) plays a decisive role in a number of applied problems, ranging from {\it machine learning} (see \cite{ST} and references therein) to the {\it Crystallization Problem} (CP). We refer the reader to \cite{DPS} for the relation between the  EIP in the triangular lattice and the CP with respect to a two-body interatomic energy characterized by the {\it sticky-disc} interaction potential (see \cite{Heitmann-Radin80,Radin81} for more details).

Our main objective is to prove that any minimizer of the EIP$_n$  --after a suitable translation--  
differs from a fixed {\it cubic configuration} 
\be{wulffside}
W_n:= [0,\ell_n]^3\cap\Zz^3 \quad\mbox{with}\quad  
\ell_n:=\lfloor\sqrt[3]{n}\rfloor.
\ee 
(with respect to the cardinality of their symmetric difference) by at most 
\begin{equation}\label{estimation}
K\, n^{3/4}  + {\rm o}(n^{3/4})
\end{equation}
elements of $\Zz^3$ for some {\rm universal} positive constant $K>0$ (see Theorem \ref{main1}), and to show that this estimate is sharp for infinitely many $n$. In particular, the exponent $3/4$ of the leading term cannot be lowered in general (see Theorem \ref{main2}).  In the following we refer to the cubic configuration $W_n$ as the {\it Wulff shape} because of the analogy to the crystallization problem. 

We first show that \eqref{estimation} is an {\it upper bound} for every  minimizer of EIP$_n$.

\begin{theorem}[Upper Bound]\label{main1} 
There exists  a constant  $K_1>0$  independent of $n$ \EEE
such that  
\be{t-f-law}
 \min_{a \in \Zz^3} \#(M_n\triangle (a + W_n) ) \leq \PPP K_1 \EEE n^{3/4} + {\rm o}(n^{3/4})
\ee
for every $n\in\Nz$ and every minimizer $M_n$ of the EIP$_n$.
 \end{theorem}

 Our second result shows that the exponent $3/4$ in \eqref{estimation} is optimal.

 \begin{theorem}[Lower Bound]\label{main2} 
There exists a sequence of minimizers $M_{n_i}$ with a diverging number $n_i\in\Nz$ of particles \PPP such that
\be{t-f-law2}
 \min_{a \in \Zz^3} \#(M_{n_i}\triangle (a + W_{n_i}) ) \ge K_2  {n_i}^{3/4} + {\rm o}(n_i^{3/4})
\ee
for some constant $K_2>0$ \emph{(}not depending on $n_i$\emph{)}. \EEE
 \end{theorem}

 We will prove Theorem \ref{main1} in Section \ref{sec:ProofThm1} and Theorem \ref{main2} in Section \ref{sec:ProofThm2}. By setting $K := \limsup_{n\to \infty} n^{-3/4} \max_{M_n} \min_{a \in \Zz^3} \#(M_n\triangle (a + W_n) )$, 
 where the maximum is taken among all configurations $M_n$ that are EIP$_n$ minimizers, by Theorems \ref{main1} and \ref{main2} we have $K\in(0,+\infty)$. We see that it is possible to choose $K_1 = K_2 = K$ in \eqref{t-f-law} and \eqref{t-f-law2}.  \EEE




 
These results for $\Zz^3$ are the first ones related to
fluctuations of minimizers in  three dimensions. Analogous results in
two dimensions have been established in \cite{DPS, Schmidt} for the
triangular lattice (see also \cite{Yeung-et-al12}), in \cite{MPS,MPS2}
for the square lattice, and finally in \cite{DPS2} for the hexagonal
lattice. The methods have been based on rearrangements techniques
\cite{Schmidt} and on the isoperimetric characterization of the
minimizers (with respect to suitable notions of perimeter $P$ and area
$A$ of configurations) which also allows to find the optimal constants
for  relations of the type of \eqref{t-f-law}  (see \cite{DPS2,DPS}).

\section{Mathematical Setting}\label{sectionsetting}

In this section we introduce the main definitions and notations used throughout the paper.

We first recall a useful characterization of EIP$_n$ minimizers that we shall often exploit.  The number of (unit) {\it bonds} of a configuration $C_n\in\mathcal{C}_n$ is 
\[b(C_n):=\frac12 \,\#\{(x,y)\in C_n \times C_n: |x-y|=1\}.
\]
Then the elementary relation $\#\Theta(C_n)+2b(C_n)=6n$ shows that $C_n$ is a minimizer for the EIP$_n$ if and only if it maximizes the number of unit bonds.

We  also introduce the $2$-dimensional  analogon of \eqref{eip3} which we denote here as EIP$^2_d$ for  $d\in\Nz$, i.e.,  
\begin{equation}\label{eip2}
\eta_d:=\min_{E_d\in\mathcal{C}^{2}_d} \#\Theta_{2}(E_d), 
\end{equation}
where $\mathcal{C}^{2}_d$ is the family of  subsets of the square lattice $\Zz^2$ with $d$ elements  
and 
\begin{equation}\label{edge2}
\Theta_2(E_d):=\{ (x,y)\in\Zz^2\times\Zz^2 \ : \ \textrm{$|x-y|=1$, $x\in E_d$ and $y\in \mathbb{Z}^2\setminus E_d$} \}.
\end{equation}
We recall from \cite{MPS} that $E_d$ solves \eqref{eip2} if and only if the number of unit bonds of $E_d$ (i.e., $ \tfrac12 \,\#\{(x,y)\in E_d\times E_d: |x-y|=1\}$) is equal to $\lfloor 2d-2\sqrt{d}\rfloor$. 
 This is equivalent to $\#\Theta_{2}(E_d) = 4d - 2\lfloor 2d-2\sqrt{d}\rfloor$, i.e., to 
\be{Theta-2-char} 
\#\Theta_{2}(E_d) = 2\lceil 2\sqrt{d} \rceil. 
\ee

We also recall from \cite{MPS} that for any $d\in \Nz$ there exists a minimizer $D_d$ of EIP$^2_d$ of the type
\be{daisy}
D_d:=R(s,s')\cup L_e
\ee
 for some $s,s'\in\Nz$ and $e\in\Nz\cup\{0\}$ such that  $s'\in\{s,s+1\}$, $s\cdot s'+e=d$, and $e< s'$, where
\be{rectangle} 
R(s,s'):=\Zz^2\cap\left([1,s]\times[1,s']\right)
\ee
 and  
\be{line}
\PPP L_e :=  \begin{cases}\UUU\Zz^2\PPP\cap\left( (0,e]\times\{s'+1\} \right)\quad&\textrm{if $s'=s$},\\
 \UUU\Zz^2\PPP\cap\left(\{s+1\}\times(0,e] \right)\quad&\textrm{if $s'=s+1$}. \EEE
 \end{cases}
\ee
\PPP Notice that if $e=0$, then $L_e=\emptyset$. \EEE
We refer to these $2$-dimensional minimizers as {\it daisies} (as already done in \cite{MPS}) and to the integers $s-1$ (resp. $s'-1$) as the minimal (resp. maximal) {\it  side length} of the rectangle \eqref{rectangle} of the daisy. 

Let us also introduce the notion of  {\it minimal rectangle} associated to a $2$-dimensional configuration $C_n$ in $\UUU\Zz^3\EEE$.

 
\EEE

\begin{definition}\label{minrectangle}
Given a configuration 
$$
C_n\subset\{k_1 \eone+k_2\etwo+ z \ethree\,\,:\quad\textrm{$k_i\in\Zz$ for $i=1,2$}\}
$$
 for some $z\in\Zz$, we denote by $R(C_n)$ the closure of the {\it
   minimal rectangle} containing $C_n$, i.e., the \UUU minimal \EEE
 rectangle \UUU $R$ \EEE with respect  to set inclusion in 
 $$
 \Rz^2_z:=\{(z_1,z_2,z_3)\in\Rz^3\,:\, z_3=z\}
 $$
\UUU with \EEE sides parallel to $\ei$ for $i=1,2$, and that satisfies $ C_n\subset R$. 
\end{definition}

Moving ahead to $3$-dimensional configurations in $\UUU\Zz^3\EEE$ we
introduce here a discrete rearrangement procedure, which we call {\it
  cuboidification}. Notice that the cuboidification is the 
3-dimensional  analogue  of the 2-dimensional  rearrangement introduced in \cite{MPS} and denoted {\it rectangularization} (see also \cite{ Bollobas, Harary, Harper}), even though here we define such rearrangement only for EIP$_n$ minimizers and not for a general configuration.  To this end, let us introduce for every $z\in\Zz$ 
 the notion of $z$-levels of a configuration $C_n\subset\UUU\Zz^3\EEE$ in the direction $i=1,2,3$, i.e., the $2$-dimensional configurations defined by
\begin{align*}
C_n(z, \cdot,\cdot)
&:=C_n\cap\{z \eone+k_2\etwo+ k_3 \ethree\,\,:\quad\textrm{$k_i\in\Zz$ for $i=2,3$}\}, \nonumber\\
C_n(\cdot, z,\cdot)
&:=C_n\cap\{k_1 \eone+z\etwo+ k_3 \ethree\,\,:\quad\textrm{$k_i\in\Zz$ for $i=1,3$}\}, \quad \textrm{and} \\ 
C_n(\cdot,\cdot, z)
&:=C_n\cap\{k_1 \eone+k_2\etwo+ z \ethree\,\,:\quad\textrm{$k_i\in\Zz$ for $i=1,2$}\}, 
\end{align*}
respectively. \PPP In the following, we  also denote by $Q(a_1,a_2,a_3)$ for some $a_i\in\Nz$ the closed  cuboid 
$$Q(a_1,a_2,a_3):=\UUU [1,a_1] \times [1,a_2]\times [1,a_3]. \EEE $$
\EEE Furthermore, we define the {\it 3-vacancies} of a configuration
$C_n\subset\UUU\Zz^3\EEE$ as the elements of $\UUU\Zz^3\EEE\setminus
C_n$ that would activate three bonds if added to $C_n$, i.e., those
elements of $\UUU\Zz^3\EEE\setminus C_n$ which have a distance $1$ to
\UUU exactly \EEE  three different elements of $C_n$. 

\begin{definition} \label{cuboidification} 
We define the \emph{cuboidification} $\mathcal{Q}(M_n)$ (in the direction $\ethree$) of a minimizer $M_n$ of the EIP$_n$  as the configuration  resulting from rearranging the particles of $M_n$ according to the following three steps.

\begin{itemize}
\item[(i)]  For every $z\in\Zz$, let $d_z:=\#M_n(\cdot,\cdot,
  z)$ and 
  \EEE consider the $2$-dimensional daisy $D_{d_z}$ (which has been defined in \eqref{daisy}), order the elements of the family $( D_{d_z} )_{d_z\neq0}$ decreasingly with respect to their cardinality, say $\displaystyle (D^{(k)})_{k=1,\dots,f}$ with $f:= \#\{ z \in \Zz \,:\,d_z\neq0 \}$, and consider the configuration $M'_n$ characterized by
$$
M'_n(\cdot,\cdot,k)=D^{(k)}+k\ethree
$$
  for $k=\{1,\ldots, f\}$ and  $M_n'(\cdot,\cdot,k)=\emptyset$ if
  $k\notin \{1,\ldots,f\}$, see Figure \ref{fig:M1}. 

By \eqref{daisy} there exist $s,s'\in\Nz$ and $e\in\Nz \cup 
\{0\}$ with $s\cdot s'+e=d$, $s'\in\{s,s+1\}$, 
and $e< s'$, such that $D^{(1)}= R(s,s') \cup L_e$ for $R(s,s')$ and $L_e$ defined as in \eqref{rectangle} and \eqref{line}, respectively.

It is clear that $M_n'$ is still an EIP$_n$ minimizer. Also, if $f\le 2$ the cuboidification algorithm ends here. Otherwise, we proceed to the next steps.

\begin{figure}[ht]
  \centering
  \includegraphics[width=65mm]{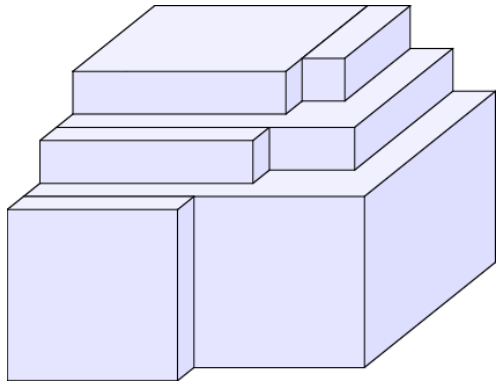}
  \caption{Configuration $M'_n$. A caveat: in favor of illustrative clarity,
    proportions in this and the following 
    figures do not correspond to the actual ones of a ground state.}
  \label{fig:M1}
\end{figure}

\item[(ii)] Consecutively  move the elements from $M'_n(\cdot,\cdot,f)$ with at most 3 bonds (there is always at least  one of them)  to fill the 3-vacancies in $M'_n\setminus M'_n(\cdot,\cdot,f)$. This allows to obtain a configuration $M''_n$ whose levels $M''_n(\cdot,\cdot,k)$  for $k=2,\dots,f-1$ are rectangles (with possibly the extra segment $L_e+k\ethree$), i.e., 
\begin{equation}\label{minus}M''_n(\cdot,\cdot,k)\setminus(L_e+k\ethree)=R(a_k,b_k)\end{equation}
 for some $a_k$, $b_k\in\Nz\cup\{0\}$ which are decreasing in $k$,
 see Figure \ref{fig:M2}. \EEE We
 also notice that the \UUU $(z=1)$-level \PPP remains unchanged, i.e., 
$$M''_n(\cdot,\cdot,1)=M'_n(\cdot,\cdot,1)=D^{(1)}+\ethree,$$
and we assume, without loss of generality, that $M''_n(\cdot,\cdot,f)$ is a daisy. \EEE

We stress that three-vacancies, if any, are filled one by one, first at the level $z=2$, 
 and then at the levels $z=3,\ldots f-1$ (in this order), in such a way that all the $z$-levels up to $f-1$ are still daisies. Moreover, each of these daisies is either coinciding with $M''_n(\cdot,\cdot, 1)$, or it is a rectangle: that is, if \eqref{minus} holds and $a_k<s$ or $b_k<s'$, then actually $M_n''(a_k,b_k)=R(a_k,b_k)$. 
 $M''_n$ is also an EIP$_n$ minimizer.

\begin{figure}[ht]
  \centering
  \includegraphics[width=65mm]{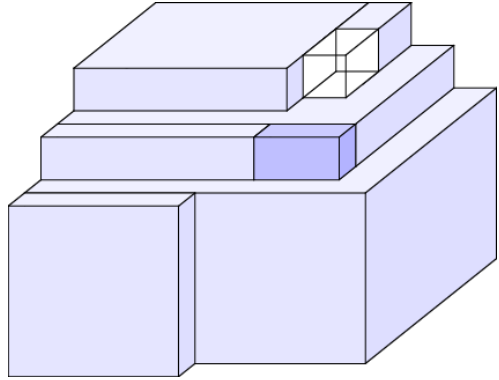}
  \caption{Configuration $M''_n$. }
  \label{fig:M2}
\end{figure}

\item[(iii)] \PPP We now construct a configuration $M'''_n$ by iteratively performing the following procedure $\mathcal{P}_k$ for $k=2,\dots,f-1$. 
The procedure $\mathcal{P}_k$ consists of performing the following two substeps: 
\begin{itemize}
\item[1)] If $a_k= s$ we directly pass to substep 2). If instead $a_k< s$, then we move an entire external edge from the $f$-level (an $f$-level edge smaller or equal to $b_k$ exists by Step (ii)), attach it at the $k$-level so that each of its atoms is bonded both to an atom that was already at the $k$-level and to one atom at the $(k-1)$-level, and, if any 3-vacancies at the $k$-level appeared, then we repeat Step (ii) in order to fill them. This can be performed so that the $k$-level of the obtained configuration is $R(a_k+1,b_k)$. 
By iterating this Substep $a=s-a_k$ times, the $k$-level of the resulting configuration  is $R(s,b_k)$.
\item[2)] If $b_k= s'$, then the procedure $\mathcal{P}_k$ is finished. If instead $b_k< s'$, then we move an entire external edge from the $f$-level, attach it at the $k$-level so that each of its atoms is bonded both to an atom that was already at the $k$-level and to one atom at the $(k-1)$-level, and possibly remove any 3-vacancies by repeating Step (ii). 
We then iterate this Substep $b=s'-b_k$ times, so that the $k$-level of the final configuration is $R(s,s')$. 
\end{itemize}
\EEE For each $k=2,\ldots, f-1$, the output of this step is a $k$-level of the form $R(s,s')$, plus possibly an    extra-line $L_e+k\ethree$ (which can be there only if $a_k=s$ and $b_k=s'$, that is, only if the procedure $\mathcal{P}_k$ is empty), see Figure
\ref{fig:M3}. \EEE

\begin{figure}[h]
  \centering
   \includegraphics[width=65mm]{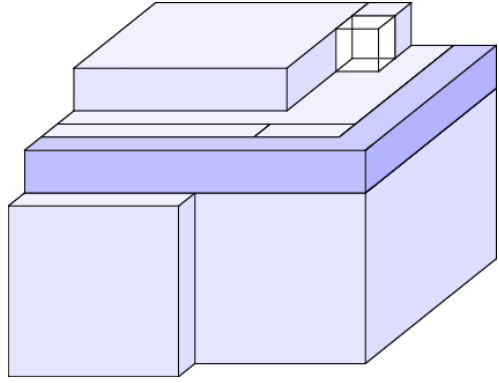}
  \caption{Configuration $M'''_n$. }
  \label{fig:M3}
\end{figure}

We notice that, if we denote by $\mathcal{P}_k(M''_n)$ the configuration obtained by iteratively performing $\mathcal{P}_i$ for $i=2,\dots,k$, we have that
$$\mathcal{P}_k(M''_n)(\cdot,\cdot,i)\setminus(L_e+i\ethree)=R(s,s')$$
for every $i=2,\dots,k$, and that 
$$M'''_n:=\mathcal{P}_{f-1}(M''_n)=(\UUU\Zz^3\PPP\cap Q(s,s',f-1))\cup F_1\cup F_2$$ 
where $F_1:=M'''_n(\cdot,\cdot,f)$ is rearranged  as a daisy \PPP and 
$$
F_2:=\begin{cases}M'''_n(\cdot,s+1,\cdot)\quad\textrm{if $s'=s$},\\
M'''_n(s+1,\cdot,\cdot)\quad\textrm{if $s'=s+1$}.
\end{cases}
$$
We notice  that $F_2\setminus F_1$ is a rectangle $R(e,s'')$ for some $e\in\{0,\ldots, s'-1\}$, $s''\in \{1,\ldots f-1\}$  and that, if $e=0$, then $F_2=\emptyset$.
 Without loss of generality we assume that
 $F_2=M'''_n(s+1,\cdot,\cdot)$, as we can move the whole $F_2$ (hence
 also the extra-line of the daisy $F_1$, if such line is contained in
 $F_2$\PPP) on that side of $Q(s,s',f)$ \PPP  since $s'\geq s$.
 Similarly, also in case $F_1 \cap F_2 = \emptyset$ moving a line of
 atoms in $F_1$ yields $F_1$ which forms a square $\Zz^3\cap (
 R(a_f,b_f)\times\{f\} )$ on the $f$
-level with an extra line of atoms in $\{a_f+1\}\times(0,e_f]\times\{f\}$. 
\end{itemize}\EEE
The configuration $\mathcal{Q}(M_n):=M_n'''$ is still an EIP$_n$ minimizer and it is the output of the cuboidification.
\end{definition}

\begin{remark}\EEE
We stress  that the recursive application of steps (ii) and (iii) in the above definition can never exhaust the upper face $M_n'(\cdot,\cdot,f)$ nor break its minimality for the two-dimensional EIP before a configuration of the form of $M_n'''$ is created, otherwise $M_n$ would not be a minimizer of the EIP$_n$.
\end{remark}

\UUU

We conclude this section with two more definitions. 
\begin{definition}\label{quasic}
We say that  a  configuration $C_n$ is {\it quasicubic}, or a {\it quasicube}, \PPP if there exist $s,s',s_3\in\Nz$ with $s'\in\{s,\PPP s+1\EEE\}$ such that (up to translation, relabeling, and reorienting the coordinate axes) 
$$C_n=\left(\UUU\Zz^3\PPP\cap Q(s,s',s_3-1)\right)\cup F^1_{d_1}\PPP\cup F^2_{d_2}\EEE $$
where $F^i_{d_i}$, $i=1,2$, are configurations with cardinality $d_i:=\#F^i_{d_i}$ such that $$F^1_{d_1}\subset \UUU\Zz^3\PPP\cap([1, s+1]\times[1,s']\times\{s_3\})$$ and 
$$F^2_{d_2}\subset \Zz^3\cap(\{s+1\}\times[1,s'-1]\times[1,s_3]).$$ 
\end{definition}

\noindent We observe that  by Definition \ref{cuboidification} the cuboidification $\mathcal{Q}(M_n)$ of an EIP$_n$ minimizer $M_n$ is a quasicube with \PPP $s_3=\#\{z\, :\, M_n(\cdot,\cdot, z)\neq\emptyset\}$ and $s-1$ the smallest side length of the rectangle $[1,s]\times[1,s']$ of the daisy with $\max_z \#M_n(\cdot,\cdot, z)$ elements.

Finally, we define the  minimal cuboid of a configuration $C_n\subset\UUU\Zz^3\EEE$.

\begin{definition}\label{mincube}
Given a configuration $C_n\subset\UUU\Zz^3\EEE$ we denote by $Q(C_n)$ the closure of the {\it minimal rectangular cuboid} containing $C_n$, i.e., the smallest rectangular cuboid $Q$ with respect to set inclusion in $\Rz^3$ that has sides parallel to $\ei$ for $i=1,2,3$, and such that  $C_n\subset Q$. 
\end{definition}


\section{Upper bound: Proof of Theorem \ref{main1}}\label{sec:ProofThm1}

We exploit the cuboidification algorithm from Section \ref{sectionsetting} to obtain the proof of our first main result.

\begin{proofad1}
Fix $n\in\Nz$ and let $M_n$ be a minimizer  of  EIP$_n$. 
 In the following, without loss of generality (up to a translation and rotation of the coordinate system), we assume that 
\be{mincubeMn}
Q(M_n)=Q(\ell_1\PPP+1\EEE,\ell_2\PPP+1\EEE,\ell_3\PPP+1\EEE)
\ee
for some $\ell_1,\ell_2,\ell_3\in\Nz\cup\{0\}$ with
\be{sides}
0\leq\ell_1\leq \ell_2 \leq \ell_3,
\ee
 where $Q(M_n)$ is  the  minimal cuboid of $M_n$ (see Definition \ref{mincube}) \PPP and $\ell_i$, $i=1,2,3$ its side lengths. \EEE
Notice also that $Q(W_n)=Q(\ell_n\PPP+1\EEE,\ell_n\PPP+1\EEE,\ell_n\PPP+1\EEE)\PPP-(1,1,1)\EEE$ for $\ell_n$ defined in \eqref{wulffside}.

We now claim that there exists a constant $K>0$ (which does not depend on $n$ and $M_n$)  such that
\be{claimupper}
\displaystyle\max_{i=1,2,3} |\ell_i-\ell_n|\leq K n^{1/12} +o(n^{1/12}).
\ee
\UUU Once this is proved, \EEE Theorem \ref{main1} follows since
$\ell_n=n^{1/3} +{\rm O}(1)$ by \eqref{wulffside} and hence  each
external face of
$Q(\ell_1\PPP+1\EEE,\ell_2\PPP+1\EEE,\ell_3\PPP+1\EEE)$ intersected
with $\UUU\Zz^3\EEE$ has cardinality $n^{2/3}+{\rm o}(n^{2/3})$.   Thus we obtain 
$$ \min_{a \in \Zz^3} \#(Q(M_n)\triangle (a + W_n) ) \leq 3K n^{3/4} + {\rm o}(n^{3/4}). $$ 
Since by \eqref{claimupper} moreover 
\begin{align*}
  \#( Q(M_n) \setminus M_n ) 
  &= (\ell_1 + 1) (\ell_2 + 1) (\ell_3 + 1) - n \\ 
  &\le (n^{1/3} + K n^{1/12} + {\rm O}(1))^3 - n 
   = 3K n^{3/4} + {\rm o}(n^{3/4}), 
\end{align*}
\eqref{t-f-law} follows.

In order to prove \eqref{claimupper}  we proceed in 5 steps. 

\subsection*{Step 1.} In this step we show that by rearranging the
elements of $M_n$ we can construct another minimizer $\overline{M}_n$
of the EIP$_n$ \PPP that is  quasicubic, i.e., there exists
$\ell,\ell',\UUU \ell_3 \EEE\in\Nz\cup\{0\}$ with $0\leq \ell\leq \ell_3$, $\ell'\in\{\ell,\ell+1\}$,  
and configurations $F^1_{d_1}$ and $F^2_{d_2}$ as in Definition \ref{quasic} such that  \EEE
\begin{equation}\label{quasicubicminimizer}
\overline{M}_n=\left(\UUU\Zz^3\EEE\cap Q(\ell\PPP+1\EEE,\ell'\PPP+1\EEE,\PPP\ell_3\EEE)\right)\cup F^1_{d_1}\cup F^2_{d_2}.
\end{equation}
This assertion follows by choosing $\overline{M}_n=\mathcal{Q}(M_n)$ and by observing that  \eqref{quasicubicminimizer} is satisfied with $\ell$ being the \PPP smallest side  length  of the rectangle of the daisy $D_{m}$ (see \eqref{daisy}) where $m$ is the maximal cardinality of the $z$-levels $M_n(\cdot,\cdot, z)$ of $M_n$  for $z=1,\dots,\ell_3+1$, i.e.,
$$ m:=\max_{z=1,\dots,\ell_3+1} \#M_n(\cdot,\cdot, z). $$

We notice that 
\be{claim}
\ell_3\geq\ell.
\ee 
Indeed, by \eqref{mincubeMn} we have $m\le (\ell_1+1)(\ell_2+1)$. On the other hand, by definition of daisy, since $\ell$ is the minimal side length of the rectangle of the daisy $D_m$, it clearly satisfies $\ell+1\le \sqrt{m}$. Therefore
 \begin{equation}\label{ellbound}
 \ell+1\le \sqrt{(\ell_1+1)(\ell_2+1)},
 \end{equation}
  and with \eqref{sides} we obtain \eqref{claim}.
 \PPP

\subsection*{Step 2.} 

\PPP
We now further rearrange $\overline{M}_n$ to ``get rid of'' the face
$F^2_{d_2}$ \UUU and obtain \PPP a new EIP$_n$ minimizer which we denote $\overline{\overline{M}}_n$. To this end, recall from Definition \ref{quasic} and Definition \ref{cuboidification} that
$$
F^1_{d_1}=\Zz^3\cap\left(([1,a_f]\times[1,b_f]\times\{\ell_3+1\}) \cup(\{a_f+1\}\times(0,e_f]\times\{\ell_3+1\}) \right)
$$
for some $a_f\in\{1,\ldots, \ell+1\}$, $b_f\in\{a_f,a_f+1\}$, $e_f\in \{0,\ldots, b_f-1\}$ and that
 $F^2_{d_2}\setminus F^1_{d_1}$ is a rectangle $R(e,s'')$ with 
  $e \in\{0,1,\ldots,\ell'\}$ \PPP and $s'' \in \{1,\ldots,\ell_3 \}$\PPP \PPP.  
If $e=0$ we set $\overline{\overline{M}}_n:=\overline{M}^1_n$, where
 \be{M1}
 \overline{M}^1_n:=\overline{M}_n=(\UUU\Zz^3\PPP\cap Q(\ell+1,\ell'+1,\ell_3))\cup F^1
 \ee
 for $F^1:=F^1_{d_1}$  $=\overline{M}_n(\cdot,\cdot,\ell_3+1)$.

 For $e>0$  and $F^1_{d_1}\cap F^2_{d_2}=\emptyset$  we define $\overline{\overline{M}}_n$ by distinguishing 3 cases: $a_f\ge e\vee s''$, $a_f<e$ and  $a_f<s''$. Later we will see how to treat the case $F^1_{d_1}\cap F^2_{d_2}\neq\emptyset$ \PPP
\begin{itemize}
\item[1.] If $a_f\ge e\vee s''$, \PPP
  then we move $F^2_{d_2}$ on top of $F^1_{d_1}$, and
  we consider the cuboidification of such configuration, \UUU which
  has the form \PPP
 \be{M2}
 \overline{M}^2_n:=(\UUU\Zz^3\PPP\cap Q(\ell+1,\ell'+1,\ell_3+1))\cup F^2
\ee
for some $F^2:= \overline{M}^2_n(\cdot,\cdot,\ell_3+2)$ (which we considered rearranged as a daisy). We set $ \overline{\overline{M}}_n:=\overline{M}^2_n$. 
\item[2.] Let $a_f<e$. We can assume without loss of generality that $s''=\ell_3$. In fact, if \PPP
  $s''<\ell_3 $, then we perform for $j=1,\dots, \ell_3 -s''$ the following transformation $\mathcal{T}^1_j$: Move an edge with length  less than $e$ from $F^1_{d_1}$ onto $F^2_{d_2}$, so that $F^2_{d_2}$ becomes a rectangle $R(e,s''+j)$ after removing \PPP (as done in Step (ii) of Definition \ref{cuboidification}) any 3-vacancy which might have been created.
The obtained configuration is 
$$M^{b}_n:=(\Zz^3\cap Q(\ell+1,\ell'+1,\ell_3 ))\cup F^{b}_1\cup F^{b}_2$$ 
where $F^{b}_1:=M^b_n(\cdot,\cdot,\ell_3+1)$ and $F^{b}_2$ is the rectangle $R(e,\ell_3 )$. We then iterate the following transformation $\mathcal{T}^2_p$: for every element $p\in B_e$ where
\be{missingbasis}
B_e:=\Zz^3\cap(\{\ell+2\}\times(e, \ell'+1]\times\{1\})
\ee
\UUU remove an \PPP edge of $F^{b}_1$, \UUU attach \PPP it to
$F^{b}_2$ by first rotating it in order to \UUU make it \PPP parallel
to  $\ethree$ and then translating it in such a way that one of its
endpoints coincides with $p$, and then \UUU remove \PPP  all
3-vacancies possibly created as in Step (ii) \UUU of Definition
\ref{cuboidification}, \PPP so that $F^{b}_2$ becomes  $R(e+1,\ell_3)$. We notice that we can perform  $\mathcal{T}^2_p$ since $a_f<e\le \ell'\le\ell+1\le \ell_3+1$ by \eqref{claim}, thus $a_f\le\ell_3$.  \PPP  The configuration obtained after performing $\mathcal{T}^2_p$ for every $p\in B_e$ is 
 \be{M3}
\overline{M}^3_n:= (\Zz^3\cap Q(\ell+2,\ell'+1,\ell_3 ))\cup F^3
\ee
where $F^3:=  \overline{M}^3_n(\cdot,\cdot,\ell_3+1)$ can be \UUU rearranged \PPP as a daisy. We set $ \overline{\overline{M}}_n:=\overline{M}^3_n$. \PPP

\item[3.] Let $a_f<s''$. \PPP 
 We first perform the transformation $\mathcal{T}^2_p$ for every $p\in B_e$ to obtain the configuration 
$$M^{c}_n:=(\Zz^3\cap Q(\ell+1,\ell'+1,\ell_3))\cup F^c_1\cup F^c_2$$ 
where $F^c_1:=M^c_n(\cdot,\cdot,\ell_3+1)$ and $F^c_2$ is a rectangle 
$R(\ell'+1,s'')$. 
\UUU Then, \PPP if $s''<\ell_3 $ we perform $\mathcal{T}^1_j$ for every  $j=1,\dots, \ell_3 -s''$ (without losing bonds since $a_f\le \ell+1\le\ell'+1$\PPP) and we obtain also a configuration of the type  $\overline{M}^3_n$.  Therefore, also in this case we set  $\overline{\overline{M}}_n:=\overline{M}^3_n$. \PPP

\end{itemize} 

 Let us now consider, for $e>0$, the case $F^1_{d_1}\cap F^2_{d_2}\neq\emptyset$, which is possible according to Definition \ref{cuboidification}. The latter definition implies in such case  $s''=\ell_3$, $a_f=\ell+1$, and
 $$
 F^1_{d_1}\cap F^2_{d_2}=\Zz^3\cap (\{\ell+2\}\times [1,\ldots, e_f] \times\{\ell_3+1\})
 $$
 with $1\le e_f\le e$. Since $e\le\ell'\le \ell+1$, we have $e_f\le \ell+1$, and thanks to \eqref{claim} we obtain $e_f\le \ell_3+1$. If $e_f\le \ell_3=s''$, we may move the $e_f$ points of $F^1_{d_1}\cap F^2_{d_2}$ to $\Zz^3\cap(\{\ell+2\}\times \{e+1\}\times[1, e_f])$ and preserve the number of bonds (since $e_f\le s''$). This produces a new configuration, with the same structure of $\overline{M}_n$, but with $F^1_{d_1}\cap F^2_{d_2}=\emptyset$, and starting from such configuration we can proceed as above with the three cases.
 Else if $e_f=\ell_3+1$, then by $e_f\le \ell+1$ and by \eqref{claim}
 we get $\ell=\ell_3$, hence $a_f=\ell+1=\ell_3+1=s''+1$. On the other
 hand, $e\le \ell'\le\ell+1=a_f\le b_f$. Therefore it is possible to
 remove the entire $F^2_{d_2}$ and place it above the rectangle
 $[1,a_f]\times[1,b_f]$ of $F^1_{d_1}$  and conclude by arguing 
 as in Case 1 above.

We observe that $\overline{\overline{M}}_n\in\{\overline{M}^i\,:\, i=1,2,3\}$ where $\overline{M}^i$ are defined for $i=1,2,3$  in  \eqref{M1},  \eqref{M2}, and  \eqref{M3}, respectively, and hence, 
\begin{equation}\label{twobar}
\overline{\overline{M}}_n:=  (\Zz^3\cap Q(a,\ell'+1,c))\cup F_d,
\end{equation}
where $a\in\{\ell+1,\ell+2\}$,  $c\in\{\ell_3,\ell_3+1\}$, and $F_d:=\overline{\overline{M}}(\cdot,\cdot,c+1)$ with $d:=\#F_d$.
\EEE
From here on, we assume that $a=\ell+1$. The rest of the proof for the case $a=\ell+2$ is essentially the same and we omit the details.

\EEE

\subsection*{Step 3} In this step we show that 
\PPP $$\ell_3-\ell=\sqrt{6}\alpha^{1/4}\ell^{1/4} +\, {\rm
  o}(\ell^{1/4}),$$
\UUU where $\alpha\in [0,1)$ \UUU is specified later on in \eqref{eq:alpha} and
depends on $\ell$ and $\ell_3$ only. 

Assume without loss of generality that $\ell_3-\ell\ge 4$. 
Then there exists $k\in \Nz$ such that \EEE
\be{divided3}
\PPP(c-1)\EEE-\ell=3k +r
\ee 
 for some $r\in\{0,1,2\}$. 
 We further rearrange $\overline{\overline{M}}_n$ in a new minimizer $\widetilde{M}_n$. In order to define $\widetilde{M}_n$ we  consider the following subsets of $\overline{\overline{M}}_n$
\begin{align}\label{extracuboids}
S_1&:=\bigcup_{z=z_1+1,\cdots,z_2} \overline{\overline{M}}_n(\cdot,\cdot,z),\nonumber\\
S_2&:=\bigcup_{z=z_2+1,\cdots,z_3} \overline{\overline{M}}_n(\cdot,\cdot,z),\nonumber\\
S_3&:=F_d\cup\left(\bigcup_{z=z_3+1,\cdots,z_4} \overline{\overline{M}}_n(\cdot,\cdot,z)\right),\nonumber\\
\quad\textrm{and}\quad R&:=\bigcup_{z=\ell+2,\cdots,z_1} \overline{\overline{M}}_n(\cdot,\cdot,z),
\end{align}
where $z_1:=\ell+1+ r$, $z_2:=\ell+1+ r+k$, $z_3:=\ell+1+ r+2k$, and $z_4:=\ell+1+ r+3k= c$. 
Notice that the configuration   
$$G:=\overline{\overline{M}}_n\setminus\left[R\cup\left(\bigcup_{k=1,2,3} S_k\right)\right]$$
is such that $G(\cdot,\cdot,z):=R(\ell\PPP+1\EEE,\ell'\PPP+1\EEE)$ for all $\PPP1\EEE\leq z\leq\ell\PPP+1\EEE$ (and $G(\cdot,\cdot,z):=\emptyset$ for $z>\ell\PPP+1\EEE$). 
\EEE

We then define  $\widetilde{M}_n$ as the configuration resulting by performing the following transformations:
\begin{itemize}
\item[1.] Move  $S_3$ altogether in such a way that each element on
  $\overline{\overline{M}}_n(\cdot,\cdot,z_3+1)$ loses its bond with
  $\overline{\overline{M}}_n(\cdot,\cdot,z_3)$ and gains a bond
  with $$G\cap\{k_1\eone + \etwo + k_3 \ethree\,\,:\quad\textrm{$k_i\in\Zz$ for $i=1,3$}\}.$$
Note that this first rearrangement \UUU does not \EEE change the total number of bonds of the minimizer $\overline{\overline{M}}_n$. 

\item[2.] We then move  $S_2$ altogether in such a way that each
  element on $\overline{\overline{M}}_n(\cdot,\cdot,z_2+1)$ loses its
  bond with $\overline{\overline{M}}_n(\cdot,\cdot,z_2)$ and gains a
  bond with $$G\cap\{\eone + k_2\etwo+ k_3 \ethree\,\,:\quad\textrm{$k_i\in\Zz$ for $i=2,3$}\}.$$
\UUU Again, this \EEE rearrangement \UUU does not \EEE  change the total number of bonds of the minimizer. 

\item[3.] Observe that, after Transformation 2 and a translation of
  $+k\eone+k\etwo$  the
  resulting configuration,   which we denote by $T_n$, is contained in the cuboid
  $
  [1,\ell+1+k]\times[0,\ell'+1+k]\times[1,\ell+1+r+k]
  $,
  ($0$ in the second factor is due to the new placement of $F_d$ after moving $S_3$). However, $T_n$ does not contain the points of 
  $V:=\Zz^3\cap (V_1\cup V_2\cup V_3)$, where 
  \[\begin{aligned}
  V_1&:=[k+1,\ell+1+k]\times[1,k]\times[\ell+2,\ell+r+k] \\
  V_2&:= [1,k]\times[k+1, \ell+k+1]\times[\ell+2,\ell+r+k]\\
  V_3&= [1,k]\times[1,k]\times[1,\ell+1].
  \end{aligned}\]  
  Notice that $V_1=V_2=\emptyset$ if $k=1$ and $r=0$.
   We now fill-in  the set  $V$
by  subsequently  moving edges with length $\ell$ from the side aligned \EEE in the direction $\eone$
of  $  T_n (\cdot,\cdot,z_2)$ (each  contains $\ell+1$
points). \UUU We call the resulting configuration $\widetilde{M}_n$ and we denote its (remaining) upper face $\widetilde{M}_n(\cdot,\cdot,z_2)$ by $\widetilde{F}_m$ with 
\be{a}\begin{aligned}
 m:=\#\widetilde{F}_m&=(\ell+1)(\ell'+1)-\#V \\&=(\ell+1)(\ell'+1)-(\ell+1)[k^2+2k(r+k-1)].
\end{aligned}\ee
 Notice that in all the steps the total number of bonds of the
 configuration  remains the same as in $M_n$, and so \UUU along these
 transformations the \EEE edges with length $\ell$ \UUU are not
 exhausted before  $V$ is filled (since this would contradict minimality of $\overline{\overline{M}}_n$ for the EIP$_n$). Hence, $\widetilde{M}_n$ is an EIP$_n$ minimizer as well. \EEE
 
%
%

 \EEE
  
\end{itemize}

By  \eqref{a} we have
\begin{align}\label{abovearea}
m=&(\ell+1)(\ell'+1)-\#V=(\ell+1)(\ell'+1)-(\ell+1)\left[k^2\,+\,2k(k+r-1)\right]\nonumber\\
=&(\ell+1)(\ell'+1-(k^2+2k(k+r-1))\nonumber\\
=&(\ell+1)(\ell-3k^2-2s_1k+s_2+1)\nonumber\\
=&\ell^2 -\ell (3k^2+2s_1k-s_2-2) -3k^2-2s_1k+s_2+1
\end{align}
where $s_1:=r-1\in\{-1,0,1\}$ and $s_2:=\ell'-\ell\in\{0,1\}$. 

Furthermore, as $\widetilde{M}_n$ is a minimizer of the EIP$_n$, it is not possible to gain any bond by rearranging the elements of $\widetilde{F}_m$ over $\widetilde{M}_n\setminus\widetilde{F}_m$. Therefore, $\widetilde{F}_m$ is a minimizer (up to translation) of EIP$^2_m$ (see \eqref{eip2}). 
 By \eqref{Theta-2-char}, 
\be{aboveperimeter}
\Theta_{2}(\widetilde{F}_m) = 2\lceil 2\sqrt{m} \rceil. 
\ee
\EEE Since  by the Transformation 3 the configuration $\widetilde{F}_m$ is rectangular with side lengths $\ell$ and $\ell'-\left(k^2\,+\,2k(k+r-1)\right)$ its  edge perimeter is simply 
\begin{align}\label{aboveperimeter2}
\Theta_{2}(\widetilde{F}_m) 
&= 2(\ell + 1) +2(\ell'-k^2-4k(k+r-1)+1) \nonumber\\
&=4\ell-6k^2-4s_1k +2s_2 + 4.
\end{align}
\EEE

Therefore, by \eqref{abovearea}, \eqref{aboveperimeter}, and \eqref{aboveperimeter2}, we have that
\begin{align*}
&4\ell-6k^2-4s_1k +2s_2\\
&\qquad\qquad=  2\lceil2\sqrt{\ell^2 -\ell (3k^2+2s_1k-s_2-2) -3k^2-2s_1k+s_2+1} \rceil - 4, \EEE 
\end{align*}
which can be written as
\begin{align*}
&2\ell-3k^2-2s_1k +s_2\\
&\qquad\qquad= 2\sqrt{\ell^2 -\ell (3k^2+2s_1k-s_2-2) -3k^2-2s_1k+s_2+1} - 2 +\alpha
\end{align*}
with 
\be{eq:alpha}
\alpha:=\lceil2\sqrt{m} - 1\rceil\,-\,( 2\sqrt{m} -
1)=\lceil2\sqrt{m}\rceil\,-\,2\sqrt{m}\in[0,1).
\ee

By taking the square we obtain
\begin{align*}
&(2\ell-3k^2-2s_1k +s_2 + 2-\alpha)^2\\
&\qquad\qquad\qquad=4\ell^2 -4\ell (3k^2+2s_1k-s_2-2) -12k^2-8s_1k+4s_2+4
\end{align*}
from which it is straightforward to compute 
\be{final-ell-k-relation}
4\alpha\ell=9k^4 +12s_1k^3 +2(2s_1^2 + 3\alpha-3s_2)k^2 +4s_1(\alpha-s_2)k + (\alpha-s_2)^2-4\alpha.
\ee
We now observe that \eqref{final-ell-k-relation}  yields \EEE 
\be{higher terms}
k=\frac{\sqrt{2}\,\alpha^{1/4}}{\sqrt{3}} \ell^{1/4} +\, {\rm o}(\ell^{1/4}).
\ee
Therefore, from \eqref{divided3} and \eqref{higher terms} we obtain
\be{32b}
\ell_3-\ell=\sqrt{6}\alpha^{1/4}\ell^{1/4} +\, {\rm o}(\ell^{1/4}).
\ee

\subsection*{Step 4.} In this step we show that  (with $\ell_n$ as defined in \eqref{wulffside}) \EEE 
\be{step2}
\ell_n-\ell= \sqrt{\frac{2}{3}}\, \alpha^{1/4}\ell^{1/4} + \,{\rm o} (\ell^{1/4}).
\ee 
From \eqref{twobar} we have that 
\begin{align}\label{totalnumber}
 n &=(\ell+1)\, (\ell'+1) \,  (\ell_3 + s_3) \EEE  \,+\, d =  (\ell+1) \,(\ell'+1) \,(\ell +\ell_3-\ell  +s_3 \EEE) \,+\,d\EEE\nn\\
 &= \ell^3 \,+\, (\ell_3-\ell)\, \ell^2 \,  +\,{\rm O} (\ell^2) \EEE
\end{align}
where $s_2:=\UUU \ell'-\ell\EEE\in\{0,1\}$  and $s_3 = c - \ell_3$, \EEE since $d={\rm O}(\ell^2)$ and since $\ell_3-\ell={\rm O}(\ell^{1/4})$ by \eqref{32b}. \EEE
Then, \eqref{totalnumber} together with \eqref{wulffside} yields
\begin{align}\label{wulffsideorder}
\ell_n&=\lfloor\sqrt[3]{n}\rfloor=\left\lfloor \ell \,\sqrt[3]{1\,+\,\frac{\ell_3-\ell}{\ell}\, +\,  {\rm O} \left(\frac{1}{\ell}\right) \EEE }\right\rfloor\nn\\
&= \left\lfloor \ell \,\sqrt[3]{1\,+\,\sqrt{6}\,\alpha^{1/4}\ell^{-3/4} + \,{\rm o} (\ell^{-3/4})}\right\rfloor\nn\\
&= \left\lfloor \ell\,\left( 1 +\frac{\sqrt{6}\alpha^{1/4}}{3} \ell^{-3/4} + \,{\rm o} (\ell^{-3/4})\right)\right\rfloor,
\end{align}
 where  in the second equality we used \UUU \eqref{32b}. 
\EEE
 The assertion \eqref{step2} follows now from \eqref{wulffsideorder}, since it implies
 \be{wulffminuscube}
\ell_n-\ell=\left\lfloor \sqrt{\frac{2}{3}}\, \alpha^{1/4}\ell^{1/4} + \,{\rm o} (\ell^{1/4})\right\rfloor.
\ee
 
\subsection*{Step 5.} In this step we conclude the proof of the estimate  \eqref{claimupper}.

\UUU Let us \EEE define $\varepsilon_i\in \Rz$ such that $\ell_i=\ell(1+\varepsilon_i)$. 
We begin by observing that, as a consequence of Step 4. we have that 
 \be{ell-as-1/3}
 \ell=n^{1/3}+{\rm o}(n^{1/3}). 
\ee
Furthermore, from \UUU \eqref{32b} \EEE  it follows that
\be{eps3}
\varepsilon_3 
= \frac{\ell_3-\ell}{\ell} 
\leq \sqrt{6}\,\alpha^{1/4}\ell^{-3/4} + \,{\rm o} (\ell^{-3/4}). \EEE
\ee

By \eqref{sides} and \eqref{ellbound} we have  $\ell_2\geq\ell$. \EEE Therefore, by  \eqref{eps3} we obtain that 
\be{eps2}
0\leq\varepsilon_2\leq\varepsilon_3\leq \sqrt{6}\,\alpha^{1/4}\ell^{-3/4} + \,{\rm o} (\ell^{-3/4})
\ee
 \UUU as \EEE $\ell_2\leq\ell_3$. 
 If also $\ell_1\geq\ell$, then the same reasoning yields that
 $0\leq\varepsilon_1\leq\varepsilon_3\leq\sqrt{6}\,\alpha^{1/4}\ell^{-3/4}
 + \,{\rm o} (\ell^{-3/4})$. Therefore, it only remains to consider
 the case in which $\ell_1<\ell$ and hence, $\varepsilon_1<0$. \UUU 
	We have in such case, again by \eqref{sides} and \eqref{ellbound}, \UUU
\begin{align*}
  &\ell \leq\ell_1 \ell_2 \ \Rightarrow \ \ell^2 \leq \ell^2
  (1+\varepsilon_1)(1+\varepsilon_2) \ \Rightarrow \ 0 \leq
  \varepsilon_1+\varepsilon_2+\varepsilon_1\varepsilon_2 \nonumber\\
&\ \Rightarrow
  \ -\varepsilon_1 \leq \varepsilon_2 + \varepsilon_1\varepsilon_2
\end{align*}
so that, in particular 
\begin{equation} 
      0\leq -\varepsilon_1 \leq \varepsilon_2. \label{eps1}
\end{equation}
\EEE
Therefore, by \eqref{eps3}, \eqref{eps2}, and \eqref{eps1} we conclude that
\be{eps}
|\varepsilon_i|\leq \sqrt{6}\,\alpha^{1/4}\ell^{-3/4} + \,{\rm o} (\ell^{-3/4}).
\ee
 for $i=1,2,3$. Finally, by Step 4. and \eqref{eps} we observe that
\begin{align*}
 |\ell_i-\ell_n|&\leq|\ell_i-\ell|+|\ell-\ell_n|\nonumber\\
 &\leq \ell|\eps_i|+  \sqrt{\frac{2}{3}}\, \alpha^{1/4}\ell^{1/4} + \,{\rm o} (\ell^{1/4})\nonumber\\
 &\leq  \left(\sqrt{6}+ \sqrt{\frac{2}{3}}\right)\alpha^{1/4}\ell^{1/4} + \,{\rm o} (\ell^{1/4})
\end{align*}
  for $i=1,2,3$, 
  which in turn by \eqref{ell-as-1/3} yields estimate \UUU
  \eqref{claimupper} \EEE  with $$K:=  \left(\sqrt{6}+ \sqrt{\frac{2}{3}}\right)\alpha^{1/4}$$
  where \UUU we recall that \EEE
  $\alpha:=\lceil2\sqrt{m}\rceil\,-\,2\sqrt{m}\in[0,1)$, \UUU see
  \eqref{eq:alpha}, \EEE for $m$ given by   \eqref{a}.
\end{proofad1}

\section{Lower bound: Proof of  Theorem \ref{main2}}\label{sec:ProofThm2} 
We begin this section with an auxiliary lemma about solutions to the EIP in the two-dimensional square lattice. Indeed, for nonnegative integers $s,p,q$ (with $s>p\vee q)$, we consider configurations in $\mathbb{Z}^2$ of the form
\[
\mathcal{R}_{s,p,q}:= R(s-p-1,s)\cup L^p_{s-q},
\]
where $L^p_{s-q}:=\mathbb{Z}^2\cap (\{s-p\}\times [1, s-q])$. Note that $\#\mathcal{R}_{s,p,q}=s^2-sp-q$.
\begin{lemma}\label{auxiliarylemma}
Let $s,p,q\in \mathbb{N}\cup\{0\}$ be such that $s\ge 1$, $p< s$, $q<s$. Then $\mathcal{R}_{s,p,q}$ is an {\rm EIP}$_n^2$ minimizer (where $n=\#\mathcal{R}_{s,p,q}=s^2-sp-q$) if and only if $$4(s-q)>(p+1)^2.$$ In particular, by choosing $p=\lfloor s^{1/2}\rfloor$ and $q=\lfloor s/4\rfloor$, $\mathcal{R}_{s,p,q}$ is a {\rm EIP}$_n^2$ minimizer for any $s\ge 2$. 
\end{lemma}
\begin{proof}
We  observe that the number of unit bonds in $\mathcal{R}_{s,p,q}$ is equal to $$(s-1)(s-p-1)+s(s-p-2)+2(s-q)-1.$$
 We use the fact that EIP$_n^2$ minimizers are characterized by a number of unit bonds equal to $\lfloor 2n-2\sqrt{n}\rfloor$, as recalled in Section \ref{sectionsetting}. As a consequence, $\mathcal{R}_{s,p,q}$ is an EIP$_n^2$ minimizer if and only if 
\[
(s-1)(s-p-1)+s(s-p-2)+2(s-q)-1=\lfloor 2(s^2-sp-q)-2\sqrt{s^2-sp-q}\rfloor,
\]
which is equivalent to
$\lfloor 2s-p-2\sqrt{s^2-sp-q} \rfloor=0$, thus to $$0\le 2s-p-2\sqrt{s^2-sp-q} <1.$$ As the first inequality is obvious, $\mathcal{R}_{s,p,q}$ is an EIP$_n^2$ minimizer if and only if 
\[
2s-p-1< 2\sqrt{s^2-sp-q},
\]   
which is equivalent to $4(s-q)>(p+1)^2$, as desired.

By choosing $p=\lfloor s^{1/2}\rfloor$ and $q=\lfloor s/4\rfloor$, the latter is reduced to
\[
4s>1+2\lfloor \sqrt{s}\rfloor +4\lfloor s/4\rfloor+\lfloor\sqrt{s}\rfloor^2,
\]
which is implied by  $2s>1+2\sqrt{s}$ that is clearly true for $s\ge 2$.
\end{proof}


A straightforward consequence of Lemma \ref{auxiliarylemma} is the sharpness of the $N^{3/4}$ law in the two-dimensional square lattice, see \cite{MPS}. Indeed,
we might consider the sequence
\begin{equation}\label{sequence1}
d_s:=s^2-s\lfloor s^{1/2}\rfloor-\lfloor s/4 \rfloor,\qquad s=2,3,\ldots. 
\end{equation}
It is easy to check that $d_s$ is a strictly increasing sequence. We have $d_s=\#\mathcal{R}_{s,p,q}$ with $p=\lfloor s^{1/2}\rfloor$ and $q=\lfloor s/4\rfloor$, and $\mathcal{R}_{s,p,q}$ is an EIP$_{d_s}^2$ minimizer by Lemma \ref{auxiliarylemma}. On the other hand, 
\[
s-\lfloor d_s^{1/2}\rfloor
= \left\lceil s-s\sqrt{1-\frac{sp+q}{s^2}} \right\rceil
\ge  s-s\left(1-\frac{sp+q}{2s^2}\right)
\ge \frac{\lfloor s^{1/2}\rfloor}{2}, \EEE
\]
so that we may compare the two-dimensional Wulff shape $W_{d_s}^2:=\big[1, \lfloor d_s^{1/2} \rfloor\big]^2\cap \mathbb{Z}^2$ with  $\mathcal{R}_{s,p,q}$ and get
\[
\min_{a\in\mathbb{Z}^2}\#(\mathcal{R}_{s,p,q}\,\triangle\, (a+ W_{d_s}^2))
\ge  \frac{1}{2} \lfloor s^{1/2}\rfloor \EEE (s-\lfloor s^{1/2}\rfloor-1),
\]
for any $s\ge 2$. As \eqref{sequence1} implies $s=d_s^{1/2}+o(d_s^{1/2})$, we find
\[
\min_{a\in\mathbb{Z}^2}\#(\mathcal{R}_{s,p,q}\,\triangle\, (a+ W_{d_s}^2))\ge\frac12 d_s^{3/4}+ o(d_s^{3/4}).
\]

We now proceed with the proof of the three-dimensional counterpart of this result.
\begin{proofad2}
Let us consider the strictly increasing sequence $$n_s:=s^3+s^2-s\lfloor s^{1/2}\rfloor-\lfloor s/4 \rfloor,\qquad s=2,3,\ldots.$$
For any integer $s\ge 2$, we may consider the configuration
\be{initialminimizer}
M_{n_s}:=Q(s,s,s)\cup F_{d_s}
\ee
where we have introduced a $2$-dimensional configuration $F_{d_s}:=M_{n_s}(\cdot,\cdot,s+1)$ with $d_s:=\#M_n(\cdot,\cdot,s+1)$. More precisely, we define the top face $F_{d_s}$ as 
\be{topface}
F_{d_s}:=\left(  \Zz^3\cap\left([ 1,r]\times[1,s]\times\{s+1\}\right) \right) \cup L^{s},\qquad r:=s-\lfloor s^{1/2}\rfloor -1,
\ee
where \begin{equation}\label{defq}
L^{s}:=\mathbb{Z}^3\cap(\{r+1\},[1,s-q],\{s+1\}), \qquad q:=\lfloor s/4\rfloor.
\end{equation}
We see that
\begin{equation}\label{new1}
d_s=s^2-s\lfloor s^{1/2}\rfloor-\lfloor s/4 \rfloor<s^2
\end{equation}
 and that 
 \begin{equation}\label{sd}
 n_s=s^3+d_s
 \end{equation}
  is indeed the number of points of $M_{n_s}$. Moreover, the top face $F_{d_s}$ is an EIP$^2_{d_s}$ minimizer for any $s\ge 2$ by an application of Lemma \ref{auxiliarylemma}. This ensures minimality of $M_{n_s}$ for the EIP$_{n_s}$, for any $s\ge 2$.  We stress that
  \begin{equation}\label{recondition}
  s-r=\lfloor s^{1/2}\rfloor+1.
  \end{equation}


%
 \EEE

By  \eqref{new1} and \eqref{sd} it follows that
\be{s}
s=n_s^{1/3} + {\rm o}(n_s^{1/3}),
\ee 
 and hence
\be{r1}
r=n_s^{1/3} + {\rm o}(n_s^{1/3})
\ee
by \eqref{recondition}. Furthermore, by  \eqref{new1} and \eqref{s} we have
\be{d2}
d_s=n_s^{2/3} + {\rm o}(n_s^{2/3}).
\ee
We also refine  \eqref{s} by recalling \eqref{wulffside} and by claiming that 
\be{claim2}
s=\ell_{n_s}.
\ee
%

To prove \eqref{claim2}  we observe that $s\le \ell_{n_s}$ easily follows from \eqref{sd} and that
\begin{align}\label{claim2proof}
\ell_{n_s}&=\lfloor\sqrt[3]{n_s-d_s+d_s}\rfloor=\left\lfloor \sqrt[3]{n_s-d_s}\,\sqrt[3]{1+\frac{d_s}{n_s-d_s}}\right\rfloor\nonumber\\
&=\left\lfloor s\,\sqrt[3]{1+\frac{d_s}{s^3}}\right\rfloor\le\left\lfloor s\,\Big(1+\frac{1}{3}\frac{d_s}{s^3}\Big)
\right\rfloor=s
\end{align}
where we used  \eqref{wulffside} in the first equality, \eqref{sd} in the third one, \eqref{new1} in the last one.   The claim is proved.

We then proceed to construct another minimizer denoted by $M''_{n_s}$ by performing the following two consecutive transformations on $M_{n_s}$: 

\begin{itemize}
\item[1.] Define the integer
\be{h1}
h_1=\left\lfloor\frac{1}{3}n_s^{1/12}\right\rfloor. 
\ee 
We translate by $\eone$ the top face $F_{d_s}$ and we move altogether
the edge \linebreak $\left(\{s\}\times [1,s] 
  \times\{s\}\right)\cap\UUU\Zz^3\EEE$ to the position
$\left(\{\PPP1\EEE\}\times
  [1,s]\EEE\times\{s+1\}\right)\cap\UUU\Zz^3\EEE$. We repeat then this
procedure recursively for each line (parallel to $\etwo$) with 
$s$ 
elements of $M_{n_s}$ which is included in the set 
$$
H_1:=[s-h_1,s]\times [1,s]\times[s-h_1,s].
$$
This transformation gives the configuration $M'_{n_s}$, see
Figure \ref{fig:H}.
\begin{figure}[h]
  \centering
    \includegraphics[width=65mm]{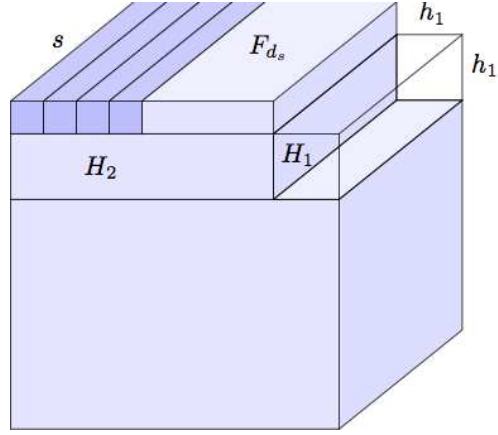}
  \caption{Configuration $M'_{n_s}$.}
  \label{fig:H}
\end{figure}
Note that $M'_{n_s}$  \EEE
 is an EIP$_{n_s}$ minimizer for large enough $s$. In fact,  the total number of moved lines  is 
$(h_1+1)^2$. 
 Hence, we can translate
$F_{d_s}$ by $(h_1+1)^2\eone$ without losing  any bond if  $ r+1+(h_1+1)^2<s-h_1$, as such condition prevents the top face to reach the points  above the `hole' $H_1$ by this translation. The latter inequality is equivalent, by \eqref{recondition}, to
\begin{equation}\label{rh1s}h_1+(h_1+1)^2<\lfloor s^{1/2}\rfloor,\end{equation} which holds true for large enough $s$ due to \eqref{claim2} and since the definition of $h_1$ entails
\begin{equation*}
(h_1+1)^2=\frac{1}{9}n_s^{1/6} + {\rm o}(n_s^{1/6}).
\end{equation*}\EEE
\smallskip

\item[2.] 
Thanks to the previous step, there exists $s_0\in\mathbb{N}$ such that, for any $s\ge s_0$, $M'_{n_s}$ is an EIP$_{n_s}$ minimizer. 
In particular, $s_0$ can be defined as the smallest integer such that
\eqref{rh1s} hold for any $s\ge s_0$. 
 For $s\ge s_0$ \EEE we move altogether the elements in the set 
$$
H_2:=[1, s-h_1-1]\times[1,s]\times[s-h_1,s+1]
$$ 
in such a way that each element of $M'_{n_s}(\cdot,\cdot,s-h_1)$ \UUU  (loses \EEE  the bond with
$M'_{n_s}(\cdot,\cdot,s-h_1-1)$ and) \UUU gets \EEE bonded with an element of  $M'_n(1,\cdot,\cdot) \setminus H_2$. \EEE We denote the resulting EIP$_{n_s}$ minimizer by $M''_{n_s}$.
\EEE
\end{itemize}

\EEE

Thanks to the two steps above,  for any $s\ge s_0$ the constructed  configuration $M''_{n_s}$ is a minimizer of the EIP$_{n_s}$ problem and  moreover we notice that
\begin{equation}\label{inclusion}
\mathbb{Z}^3\cap ([-h_1,s]\times [1,s]\times [1, s-h_1-1])\subset M''_{n_s},
\end{equation}
therefore by \eqref{wulffside}, \eqref{claim2}, \eqref{h1}, \eqref{inclusion} we conclude that
\[\begin{aligned}
\min_{a\in\mathbb{Z}^3}\#(M_{n_s}''\triangle (a+ W_{n_s}))&\ge \min_{a\in\mathbb{Z}^3}\#(M_{n_s}''\setminus (a+ W_{n_s}))\\&\ge s(s-h_1-1)(s+h_1+1-\ell_{n_s})\\&=s \left(s-\left\lfloor\frac13 \, n_s^{1/12}\right\rfloor-1\right)\, \left( \left\lfloor\frac13  n_s^{1/12}\right\rfloor+1\right).
\end{aligned}\]
Hence, 
\[
\min_{a\in\mathbb{Z}^3}\#(M''_{n_s}\triangle (a+ W_{{n_s}}))\ge\frac13
n_s^{3/4}+ {\rm o}(n_s^{3/4})
\]
follows by \eqref{s}.
\EEE
\end{proofad2}

\section*{Acknowledgements}

P. Piovano acknowledges support from the Austrian Science Fund (FWF)
project P~29681 and by the Vienna Science and Technology Fund (WWTF),
the City of Vienna, and Berndorf Privatstiftung through Project
MA16-005. \UUU U. Stefanelli is partially supported by the Vienna Science and Technology Fund (WWTF)
through Project MA14-009 and  by the Austrian Science Fund (FWF)
projects F\,65, P\,27052, and I\,2375.  \EEE


\begin{thebibliography}{99}



\bibitem{Yeung-et-al12} 
Y. Au Yeung, G. Friesecke, B. Schmidt.
Minimizing atomic configurations of short range pair potentials in two dimensions: crystallization in the Wulff-shape,
{\it Calc. Var. Partial Differential Equations}, {\bf 44} (2012), 81--100.

\bibitem{Bezrukov}
S.L. Bezrukov.
Edge isoperimetric problems on graphs, in: Graph theory and combinatorial biology (Balatonlelle, 1996).
{\it Bolyai Soc. Math. Stud.}, {\bf 7} (1999), 157--197. 


\bibitem{BL}
X. Blanc, M. Lewin. The crystallization conjecture: a review, {\it EMS Surv. Math. Sci.}, \textbf{2} (2015), 255--306.


\bibitem{Bollobas} B. Bollobas and I. Leader. Edge-isoperimetric inequalities in the grid, {\it Combinatorica}, 11 (1991), 4:299--314.




\bibitem{DPS2} E. Davoli, P. Piovano, U. Stefanelli. Wulff shape emergence in graphene,  {\it Math. Models Methods Appl. Sci.}, \textbf{26} (2016), 2277--2310. 

\bibitem{DPS} E. Davoli, P. Piovano, U. Stefanelli. Sharp $N^{3/4}$ law for the minimizers of the edge-isoperimetric problem on the triangular lattice, {\it J. Nonlinear Sci.}, \textbf{27} (2017), 627--660. 



\bibitem{Harary}
F. Harary and H. Harborth. Extremal animals, {\it J. Comb. Inf. Syst. Sci.}, 1 (1976), 1-8.

\bibitem{Harper}
L.H. Harper,
{\it Global methods for combinatorial isoperimetric problems.}
Cambridge Studies in Advanced Mathematics, 90. Cambridge University
Press, Cambridge, 2004.




\bibitem{Heitmann-Radin80}
R. Heitmann, C. Radin.
Ground states for sticky disks, {\it J. Stat. Phys.}, {\bf 22} (1980), 3:281--287.



\bibitem{Mainini-Stefanelli12}
E. Mainini, U. Stefanelli.
Crystallization in carbon nanostructures, {\it Comm. Math. Phys.}, {\bf 328} (2014), 2:545--571.

\bibitem{MPS} 
E. Mainini, P. Piovano, U. Stefanelli. Finite crystallization in the square lattice. {\it Nonlinearity},  {\bf 27} (2014), 717--737.

\bibitem{MPS2} 
E. Mainini, P. Piovano, U. Stefanelli. Crystalline and isoperimetric square configurations, {\it Proc. Appl. Math. Mech.}, {\bf14} (2014), 1045--1048. 

\bibitem{Radin81} 
C. Radin.
  The ground state for soft disks, {\it J. Stat. Phys.}, {\bf 26} (1981), 2:365--373.


\bibitem{Schmidt}
B. Schmidt. Ground states of the 2D sticky disc model: fine properties and $N^{3/4}$ law for the deviation from the asymptotic Wulff-shape. {\it J. Stat. Phys.}, {\bf 153} (2013), 727--738.



\bibitem{ST} 
N.G. Trillos, D. Slepcev.  Continuum limit of total variation on point clouds, {\it Arch. Ration. Mech. Anal.},  \textbf{220} (2016), 193--241. 


  

\end{thebibliography}
\end{document}